\documentclass[11pt]{article}

\usepackage[letterpaper,margin=1in]{geometry}
\usepackage{amsmath,amssymb,amsthm,mathtools}
\usepackage{MnSymbol}
\usepackage{nicefrac}
\usepackage{microtype}
\usepackage[round,authoryear]{natbib}
\usepackage{hyperref}
\usepackage{enumitem}
\usepackage{booktabs}
\usepackage{complexity}

\usepackage[ruled,vlined,linesnumbered]{algorithm2e}

\usepackage[table]{xcolor}
\definecolor{payoffgray}{gray}{0.92}
\definecolor{citeblue}{HTML}{5487B3}
\definecolor{linkorange}{HTML}{C58145}
\hypersetup{
    colorlinks=true,
    citecolor=citeblue,
    linkcolor=linkorange,
    urlcolor=linkorange
}

\usepackage{tikz}
\usepackage{colortbl}
\usetikzlibrary{arrows.meta}
\definecolor{BIblue}{HTML}{245878}
\definecolor{BIred}{HTML}{A54032}
\definecolor{BIgray}{HTML}{5B6268}

\newtheorem{maintheorem}{Main Theorem}
\newtheorem{theorem}{Theorem}[section]

\newtheorem{corollary}[theorem]{Corollary}

\newtheorem{example}[theorem]{Example}
\theoremstyle{definition}
\newtheorem{definition}[theorem]{Definition}
\newtheorem{assumption}[theorem]{Assumption}
\usepackage[nameinlink,noabbrev]{cleveref}
\crefname{maintheorem}{main theorem}{main theorems}
\Crefname{maintheorem}{Main Theorem}{Main Theorems}
\crefname{assumption}{assumption}{assumptions}
\Crefname{assumption}{Assumption}{Assumptions}

\usepackage{thm-restate}

\newcommand{\Acal}{\mathcal A}
\newcommand{\Mcal}{\mathcal M}
\newcommand{\Scal}{\mathcal S}
\newcommand{\supp}{\operatorname{supp}}
\newcommand{\DeltaSet}{\mathop{\rm \Delta}}
\newcommand{\defeq}{\vcentcolon=}
\newcommand{\veca}{\boldsymbol{a}}
\newcommand{\vecell}{\boldsymbol{\ell}}
\newcommand{\vecpi}{\boldsymbol{\pi}}
\newcommand{\vecx}{\boldsymbol{x}}
\newcommand{\matM}{\mathbf{M}}

\let\bm\boldsymbol

\def\va{{\bm{a}}}

\renewcommand\vec\bm

\def\cA{{\mathcal{A}}}

\def\cG{{\mathcal{G}}}

\def\cM{{\mathcal{M}}}

\let\eps\epsilon

\let\E\relax
\DeclareMathOperator*{\E}{\mathbb E}

\usepackage{xcolor}

\definecolor{briancolor}{rgb}{1.0,0.5,0.1}

\title{Normal-Form Correlation in Markov Games}
\author{%
    Ioannis Anagnostides\textsuperscript{1},
    Constantinos Daskalakis\textsuperscript{2},
    Gabriele Farina\textsuperscript{2}\\
    Noah Golowich\textsuperscript{3},
    Tuomas Sandholm\textsuperscript{1}\thanks{Additional affiliations: Strategy Robot, Inc., Strategic Machine, Inc., Optimized Markets, Inc.},
    and Brian Hu Zhang\textsuperscript{2}\\[0.75em]
    \small \textsuperscript{1}Carnegie Mellon University\\
    \small \textsuperscript{2}Massachusetts Institute of Technology\\
    \small \textsuperscript{3}University of Texas at Austin\\[0.5em]
    \small \texttt{\{ianagnos,sandholm\}@cs.cmu.edu}\\
    \small \texttt{\{costis,gfarina,zhangbh\}@mit.edu},
    \texttt{nzg@cs.utexas.edu}%
}
\date{}

\begin{document}
\pagenumbering{gobble}
\maketitle
\begin{abstract}
There has been a surge of recent work on correlated equilibrium concepts in Markov games. However, existing results focus on concepts weaker than \emph{normal-form correlated equilibria (NFCEs)}, leaving open the more challenging question of computing such equilibria, which goes back to the seminal work of Papadimitriou and Roughgarden (JACM'08). Here, we establish the first efficient algorithm for NFCEs in finite-horizon Markov games with a fixed number of players $n$. In particular, with $S$ states, horizon $H$, and at most $A$ actions per player, it computes an $\epsilon$-NFCE in time $S(AH/\epsilon)^{O(n)}$. This is the first algorithm polynomial in $1/\epsilon$ and the description of the game for NFCEs in an interesting class of problems beyond the normal-form setting. Moreover, under the usual assumption that recommendations are independent across states, we complement our positive result by showing \PPAD-completeness---that is, computational equivalence to Nash equilibria---either in many-player games or when the precision is exponentially small.

The key idea behind our approach is to run backward induction on a sequence of auxiliary stage games, but with the twist that in each step we compute a \emph{constant-expectation correlated equilibrium}. This is a natural refinement of correlated equilibrium in which the conditional expected payoff from obeying is independent of the recommendation. In fact, our reduction goes both ways, establishing an equivalence between constant-expectation CEs and NFCEs in Markov games. For a fixed number of players, we observe that a constant-expectation CE can be computed approximately by combining linear programming with suitable discretization. In contrast, it is \PPAD-hard in i) polymatrix (many-player) games at constant precision, and ii) two-player games at exponentially small precision. The latter result follows from an unexpected connection to rank-2 two-player games.
\end{abstract}

\clearpage
\pagenumbering{arabic}
\section{Introduction}

The correlated equilibrium, introduced by~\citet{Aumann74:Subjectivity}, is a foundational notion in game theory. Unlike Nash equilibria~\citep{Nash51:Noncooperative}, the standard solution concept in noncooperative games, the set of correlated equilibria is the feasible region of a linear program, thereby enjoying more favorable computational properties. Indeed, a seminal result by~\citet{Papadimitriou08:Computing} shows that a correlated equilibrium can be computed in polynomial time even in certain multi-player games of \emph{polynomial type}, meaning that the number of pure strategies of each player is polynomial in the description of the input. However, in most realistic settings, the number of pure strategies of each player is exponentially large, since a pure strategy should prescribe how to act across a range of different contingencies. The complexity of correlated equilibria in such settings has been a long-standing open problem, first highlighted by~\citet{Papadimitriou08:Computing}.

An important class of games with a large number of contingency plans is \emph{Markov} (aka. stochastic) games. There has been much recent work examining correlated equilibrium concepts in (fully observable) Markov games (discussed further in~\Cref{sec:related}). However, existing algorithmic results focus on relaxations of the \emph{normal-form correlated equilibrium (NFCE)}, most notably notions akin to extensive-form correlated equilibrium~\citep{vonStengel08:Extensive}, wherein recommendations are revealed only upon reaching an information set along the realized path. In the stronger normal-form notion, a mediator samples a complete policy profile beforehand and (privately) reveals to each player the entire policy recommendation. This means that a player may condition its deviation on recommendations at states that are never reached. In this paper, we focus on the following question.

\begin{quote}
    \centering
    \emph{Are there efficient algorithms for computing normal-form correlated equilibria (NFCEs) in Markov games?}
\end{quote}

NFCE is a suitable solution concept in systems where a mediator---a trusted third party---specifies complete policy prescriptions before deployment, with no need to conceal future recommendations or remain involved throughout execution. For example, a traffic coordinator could recommend complete routing policies to autonomous vehicles before departure, allowing each vehicle to inspect its full policy and decide whether to follow it without requiring further guidance during the trip. Moreover, determining whether efficient algorithms exist for NFCE calls for fundamentally new technical insights, which will deepen our understanding of correlated equilibria in dynamic environments.

\subsection{Our results}

Our main positive result is an efficient algorithm for this problem in finite-horizon Markov games with a fixed number of players $n$. 

\begin{maintheorem}
    \label{theorem:informal}
    Let $S$ be the number of states, $H$ the horizon, and $A$ the maximum number of actions per player. There is an algorithm that computes an $\epsilon$-NFCE whose running time is $S(AH/\epsilon)^{O(n)} \poly(L)$, where $L$ is the bit complexity of the input.
\end{maintheorem}

This shows that, for a fixed number of players, there is a fully polynomial-time approximation scheme (FPTAS) for normal-form correlated equilibrium. More broadly, beyond polynomial-type games, the best known result is a PTAS~\citep{Dagan24:External,Peng24:Fast}. \Cref{theorem:informal} gives an exponential improvement in the dependence on $1/\epsilon$ over those results, and is the first FPTAS for computing NFCE in an interesting class of games beyond the normal-form setting.

One limitation of our positive result is that it scales poorly with the number of players; the presence of multiple players is a well-known challenge in multiagent reinforcement learning (\emph{e.g.}~\citealp{Cui23:Breaking,Wang23:Breaking}). We complement~\Cref{theorem:informal} by showing that, under \emph{state independent recommendations}, it cannot be significantly improved.

\begin{maintheorem}
\label{theorem:informal-hard}
Computing an $\eps$-NFCE with independent state recommendations is \PPAD-hard even with horizon $H=2$, in each of the following cases:
\begin{itemize}
    \item[(i)] $(n+1)$-player succinct games, two actions per player at each state, and $\eps=\Theta(1/n)$.
    \item[(ii)] Three-player games, a constant number of states, and exponentially small $\eps$.
\end{itemize}
\end{maintheorem}

That is, computing such an NFCE is as hard as finding a Nash equilibrium in general-sum games.

\subsection{Technical overview}

We now highlight the key technical insights of our approach.

\paragraph{Failure of backward induction.} The most natural approach that comes to mind is backward induction, whereby each continuation state is assigned its expected correlated equilibrium payoff over the local recommendations. However, the crux is that in an NFCE, each player observes its entire recommended policy in advance, so recommendations at future states may reveal conditional continuation payoffs that differ from the average over the local recommendations. This means that a player can profitably deviate by altering an earlier action to avoid a continuation state where its recommendation there signals a lower payoff.

More concretely, we consider a two-stage game in which Player 1 initially chooses between a terminal payoff of \(\nicefrac{5}{4} \) and a continuation state where the players play the game ``Bach or Stravinsky,'' with payoffs $(2,1)$ and $(1,2)$ on the diagonal and $(0,0)$ otherwise (\Cref{fig:failbackward}). Mixing equally between the profiles on the diagonal is a CE of the continuation game and gives Player 1 an expected payoff of \(\nicefrac{3}{2}\). As a result, backward induction based on this CE prescribes entering that state, because $\nicefrac{3}{2} > \nicefrac{5}{4}$. However, the full policy recommendation reveals which action profile should be played in the ``Bach or Stravinsky'' game. In particular, when Player 1 is recommended the action yielding payoff $2$ ($\mathsf{B}$), following the prescribed policy is optimal. On the other hand, when recommended the action yielding payoff $1$ ($\mathsf{S}$), it can instead choose the terminal payoff of $\nicefrac{5}{4}$, which is a profitable deviation. So, the correlated distribution resulting from backward induction is not an NFCE.

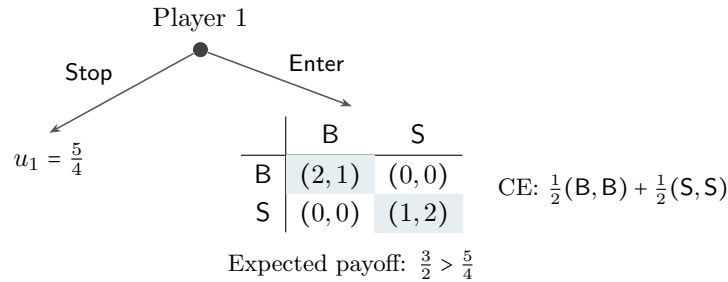
\begin{figure}[h]
    \centering
    \begin{tikzpicture}[
  font=\small,
  >={Stealth[length=1.5mm]},
  edge/.style={->,draw=black!65,line width=.6pt},
  edge label/.style={font=\footnotesize,inner sep=3pt}
]
\definecolor{CEshade}{HTML}{E7EEF2}

\node[circle,fill=black!75,inner sep=2pt,
      label=above:{Player 1}]
  (root) at (0,0) {};

\node (stop) at (-2,-1.45)
  {$u_1=\frac54$};

\node[align=center] (game) at (2,-1.65) {%
  \renewcommand{\arraystretch}{1.2}%
  $\begin{array}{c|cc}
       & \mathsf{B} & \mathsf{S} \\ \hline
     \mathsf{B} & \cellcolor{CEshade}(2,1) & (0,0) \\
     \mathsf{S} & (0,0) & \cellcolor{CEshade}(1,2)
   \end{array}$%
};

\draw[edge] (root)
  -- node[edge label,above left] {$\mathsf{Stop}$}
  (stop.north);

\draw[edge] (root)
  -- node[edge label,above right] {$\mathsf{Enter}$}
  (game.north);

\node[anchor=west,font=\footnotesize]
  at ([xshift=2mm,yshift=-2mm]game.east)
  {CE: $\frac12(\mathsf{B},\mathsf{B})
       +\frac12(\mathsf{S},\mathsf{S})$};

\node[font=\footnotesize,text=black] at (2,-2.85)
  {Expected payoff: $\frac32>\frac54$};
\end{tikzpicture}
    \caption{Backward induction prescribes $\mathsf{Enter}$ using the continuation CE's average payoff. The CE mixes equally between the two shaded profiles. Observing its full policy lets Player 1 stop whenever $\mathsf{S}$ is recommended, so the resulting distribution is not an NFCE.}
    \label{fig:failbackward}
\end{figure}

\paragraph{Constant-expectation correlated equilibria.} To address this, the key ingredient is a refinement of correlated equilibria in normal-form games, known as \emph{constant-expectation correlated equilibrium}. To our knowledge, this was first examined by~\citet{Sorin98:Distribution}, and subsequently used by~\citet{Heller12:Sequential}.\footnote{\citet{Sorin98:Distribution} used the term ``distribution equilibrium,'' but we use the more descriptive nomenclature of~\citet{Heller12:Sequential}. \citet{Heller12:Sequential} examined existence in the infinite-horizon setting, whereas we focus on computation.} In a one-shot game, a constant-expectation correlated equilibrium is a correlated equilibrium for which the conditional expected payoff from obedience \emph{does not depend on the recommended action}. Even though this is a natural notion, its complexity has not been examined before. For example, the correlated equilibrium constructed in~\Cref{fig:failbackward} which breaks backward induction does \emph{not} satisfy this property: for Player 1, the payoff conditional on observing $\mathsf{B}$ is 2, whereas the payoff conditional on observing $\mathsf{S}$ is 1. On the other hand, a Nash equilibrium clearly satisfies this requirement. However, unlike Nash equilibria, which are hard to compute~\citep{Daskalakis09:Complexity}, our key observation is that constant-expectation CEs admit an efficient algorithm (FPTAS) for a fixed number of players: we exhaustively search over each player's constant obedience payoff and solve a linear feasibility program for each candidate (\Cref{prop:grid-lp}).

\paragraph{Rectifying backward induction.} We now return to Markov games. The basic idea is to run backward induction by computing in each stage a constant-expectation correlated equilibrium. Then, the final distribution results from sampling the local recommendations independently across the states, which form pure Markovian policy profiles. This independence is crucial for the argument: after conditioning on a player's entire recommended policy, the other players' actions remain independent across states, and the environment induced by conditioning is a usual Markov decision process. The crucial role of constant expectation is that it makes conditional obedience payoffs approximately independent of the local recommendation. Together with independent sampling, this provides continuation values that remain valid after conditioning on the entire recommended policy, making backward induction sound (\Cref{lem:value-bounds}).

\paragraph{Complexity of constant-expectation CEs and implications for Markov NFCE.} We finally circle back to the proof of our main hardness result, \Cref{theorem:informal-hard}. The upshot is that, under independent state recommendations, our previous \emph{reduction goes both ways}: a constant-expectation CE reduces to an NFCE of a suitable Markov game (\Cref{thm:nfce-transfer}). Moreover, and perhaps surprisingly, we show that computing a constant-expectation CE is \PPAD-hard in the following cases.

\begin{maintheorem}
    \label{theorem:main-hardness}
    Computing a constant-expectation correlated equilibrium is \PPAD-hard in
    \begin{itemize}
        \item[(i)] binary-action $n$-player polymatrix games at constant precision, and
        \item[(ii)] two-player normal-form games at exponentially small precision.
    \end{itemize}
\end{maintheorem}
Together with our positive result, this fully characterizes the complexity of computing constant-expectation correlated equilibria, which is a natural solution concept in its own right. \Cref{theorem:main-hardness} marks a rare problem that admits an FPTAS and yet is \PPAD-hard at exponentially small approximation error. Another notable example is \emph{constant-rank} bimatrix games~\citep{Kannan10:Games,Mehta14:Constant}. Our reduction in fact starts from a rank-2 bimatrix game, establishing an unexpected connection between constant-expectation CEs and constant-rank bimatrix games. We defer the details of our hardness results to~\Cref{sec:PPAD} and additional related work to~\Cref{sec:related}.
\section{Further related work}
\label{sec:related}

\paragraph{Markov games.} 

There has been much work examining correlated equilibrium concepts in Markov games. \citet{Song22:When} and~\citet{Jin24:V} developed decentralized algorithms for learning (coarse) correlated equilibria with sample complexity that avoids an exponential dependence on the number of players. \citet{Wang23:Breaking} and ~\citet{Cui23:Breaking} extended such guarantees to settings with function approximation, while~\citet{Cai24:Near} and~\citet{Mao24:Widetilde} obtained near-optimal convergence rates in the full-feedback setting. These results concern equilibrium notions weaker than NFCE and do not guarantee obedience when each player observes its entire recommended policy before play. Our results allow arbitrary history-dependent deviations conditioned on the entire policy recommendation. On the flip side, our algorithm is centralized and model based, and is polynomial only for a fixed number of players.

A complementary line of work has established computational barriers. \citet{Daskalakis23:Complexity} showed hardness for stationary Markov coarse correlated equilibria in discounted stochastic games, while also giving efficient algorithms for nonstationary Markov coarse correlated equilibria. \citet{Farina26:Complexity} showed that hardness persists even in two-player Markov games with a single controller. These hardness results concern stationary equilibria in the discounted setting, whereas our algorithm leverages the finite horizon through backward induction. Separately, \citet{Foster23:Hardness} established hardness for computing coarse correlated equilibria represented as sparse mixtures of product policies, which rules out computationally efficient independent no-regret learning. More recently, \citet{Yorulmaz26:Independent} obtained computational lower bounds for independent learning in discounted games.

Beyond fully observable Markov games, \citet{Liu26:Partially} examined partially observable games with information sharing, obtaining quasi-polynomial guarantees under suitable information-sharing and observability assumptions. Our approach relies on full observability to solve the continuation problems separately.

\paragraph{Beyond Markov games.} As discussed earlier, existing algorithmic results focus on relaxations of the normal-form correlated equilibrium. Such relaxations can be understood through the framework of \emph{$\Phi$-equilibria}~\citep{Greenwald03:General,Stoltz07:Learning,Gordon08:No}, which has engendered much recent work~\citep{Morrill21:Efficient,Morrill21:Hindsight,Arunachaleswaran25:Swap,Daskalakis25:Efficient,Zhang25:Learning}. \citet{Anagnostides26:Complexity} showed that computing a normal-form correlated equilibrium in concave (quadratic) games is as hard as finding fixed points of contractions, which is believed to be a hard problem. The complexity, for example, in extensive-form games remains an open question. In particular, it is open whether there is an FPTAS for NFCE even in two-player (imperfect-information) extensive-form games. It is worth noting again that our approach relies crucially on full observability and does not directly extend to extensive-form games. The basic challenge in extensive-form games is that continuation problems cannot be solved independently, so there are additional obstructions in applying backward induction.

\paragraph{Constant-rank bimatrix games.} An unexpected connection we uncover in our paper is between constant-expectation CEs and Nash equilibria in constant-rank bimatrix games. This class of games was introduced by~\citet{Kannan10:Games}, motivated by applications in economics. Unlike general bimatrix games, where there is no FPTAS unless $\P = \PPAD$~\citep{Chen09:Settling}, \citet{Kannan10:Games} showed that constant-rank bimatrix games admit an FPTAS (\emph{cf.}~\citealp{Alon13:Approximate}). \citet{Mehta14:Constant} later established that, even in rank-$3$ bimatrix games, finding a Nash equilibrium is \PPAD-complete when the precision is exponentially small. In contrast, rank-$1$ games admit a polynomial-time algorithm~\citep{Adsul11:Rank} and are amenable to learning dynamics~\citep{Patris24:Learning}.

\paragraph{Other refinements of correlated equilibria.} A related refinement is \emph{information-value-free correlated equilibrium (IVFCE)}, introduced recently by~\citet{Hartline26:Computable}, which requires that each player's (unconditional) expected payoff can be matched by a fixed action chosen without observing the recommendation. In contrast, constant-expectation CE requires each player's conditional expected payoff from obedience to be independent of its recommendation. Neither property implies the other, although both hold at any Nash equilibrium. The two notions also exhibit similarities in their computational complexity. In succinct many-player games, both are \PPAD-hard to compute even at constant precision~\citep{Anagnostides26:Algorithmic}. In explicit normal-form games, however, an exact IVFCE can be computed in polynomial time~\citep{Hartline26:Computable}, whereas we show that computing a constant-expectation CE is \PPAD-hard at exponentially small precision even with two players, despite admitting an FPTAS for a fixed number of players.

\section{Preliminaries}

We first describe the basic models of normal-form and Markov games.

\subsection{Normal-form games}

A finite \emph{normal-form game} is specified by $\mathcal{G}=([n],(\Acal_i)_{i=1}^n,(u_i)_{i=1}^n)$. We make the usual normalization assumption $u_i: \Acal\to[0,1]$, where $\Acal\defeq\prod_i \Acal_i$ is the joint action space. For a (potentially correlated) distribution $\mu\in\DeltaSet(\Acal)$, we use the shorthand notation
\begin{align*}
    p_i^\mu(a_i)
    \defeq\sum_{\veca_{-i}}\mu(a_i,\veca_{-i}) \text{ and }
    w_i^\mu(a_i)
    \defeq\sum_{\veca_{-i}}\mu(a_i,\veca_{-i})u_i(a_i,\veca_{-i}).
\end{align*}
Our main notion of equilibrium is the {\em correlated equilibrium}. We will be interested in two different notions of {\em approximate} correlated equilibrium, which we now define.
\begin{definition}[Correlated equilibrium; \citealp{Aumann74:Subjectivity}]\label[definition]{def:ce}
A distribution $\mu\in\DeltaSet(\Acal)$ is a \emph{$\delta$-correlated equilibrium (CE)} if, for every player $i \in [n]$ and deviation function $\phi_i : \Acal_i \to \Acal_i$,
\[
    \E_{\veca\sim\mu}
    \bigl[u_i(\phi_i(a_i),\veca_{-i})-u_i(\veca)\bigr]
    \le\delta.
\]
It is a $\delta$-{\em conditional correlated equilibrium} if, 
for every $i$ and every pair $a_i, a_i' \in \cA_i$ with $p_i^\mu(a_i) > 0$,
\begin{equation}\label{eq:nfce}
    \E_{\va \sim \mu}\bigl[u_i(a_i', \va_{-i}) - u_i(a_i, \va_{-i}) \mid a_i] \le \delta.
\end{equation}
The latter notion of approximation is stronger: every $\delta$-conditional CE is also a $\delta$-CE, but not vice-versa. Our upper bounds will use conditional CE, and our lower bounds will use CE.
\end{definition}

% Whenever $p_i^\mu(a_i)>0$, the \emph{conditional obedience payoff} is $\mathbb{E}_{\veca\sim\mu}[u_i(\veca)\mid a_i]=w_i^\mu(a_i)/p_i^\mu(a_i)$.

\subsection{Markov games}

An $n$-player {\em Markov game} is played over stages $t=1,\ldots,H$. At stage $t$, the state belongs to a finite set $\Scal_t$. At state $s \in \Scal_t$, player $i$ has a finite action set $\Acal_{i,t,s}$. The joint action set is denoted by $\Acal_{t,s} \defeq\prod_{i=1}^n \Acal_{i,t,s}$. We will write $\veca_t\in \Acal_{t,s_t}$.  Player $i$ receives reward $r_{i,t}(s_t,\veca_t)$, and the next state is drawn from $\mathbb{P}_t(\cdot\mid s_t,\veca_t)\in\DeltaSet(\Scal_{t+1})$. (This randomization is independent conditioned on the current state and joint action.) The initial state is drawn from a distribution $\lambda\in\DeltaSet(\Scal_1)$.

A \emph{history} before stage $t$ is $h_t=(s_1,\veca_1,s_2,\veca_2,\ldots,s_{t-1},\veca_{t-1},s_t)$.  A \emph{pure (history-dependent) policy} of player $i$ maps every feasible history $h_t$ to an action in $\Acal_{i,t,s_t}$. We will use $\Pi_i$ to denote the set of pure policies of player $i$, and $\Pi := \prod_{i=1}^n \Pi_i$ to denote the set of pure policy profiles. A policy is {\em Markov} if it ignores the history before $(t, s)$. In particular, a Markov policy is given by a vector $\vecpi_i = \bigl(\pi_{i,t}(s)\bigr)_{t\in[H],\,s \in\Scal_t}$, where $\pi_{i,t}(s)\in \Acal_{i,t,s}$. For a {\em joint policy} $\vecpi = (\pi_1, \dots, \pi_n)$, we let $u_i(\vecpi)\defeq\mathbb{E}[\sum_{t=1}^H r_{i,t}(s_t,\veca_t)]$ be the expected cumulative reward. For convenience, we make the following normalization assumption.

\begin{assumption}[Normalization]\label[assumption]{ass:normalization}
For every $i,t,s,\veca$, we have $0\le r_{i,t}(s,\veca)\le 1/H$.
\end{assumption}
In particular, every cumulative payoff lies in $[0, 1]$. This normalization is without loss after rescaling.

\paragraph{Normal-form correlation.}
A finite-horizon Markov game induces a finite normal-form game, in which every (possibly history-dependent) policy $\pi_i$ is an action. Our main solution concept for Markov games will be the {\em normal-form correlated equilibrium} (NFCE), which is simply \Cref{def:ce} applied to this induced normal-form game.  As with normal-form games, we will use both the standard and conditional notions of approximate NFCE; our upper bounds will apply to the latter and lower bounds to the former.
% The normal form of the Markov game treats each pure policy as one action. In an NFCE, the mediator samples a complete policy before play and privately reveals each player's policy to that player. This is the correlated equilibrium of the induced normal-form game.
This equilibrium requirement is significantly stronger than those guaranteed by existing algorithms, which revolve around notions akin to \emph{extensive-form correlated equilibrium}, where recommendations are revealed only when the corresponding information sets are reached \citep{vonStengel08:Extensive}.

% We will use two different notions of {\em approximate} normal-form correlated equilibrium, which we now define.

% \begin{definition}[NFCE]\label[definition]{def:nfce}
% Let $\rho$ be a distribution over pure policy profiles.  The mediator draws $\vecpi\sim\rho$ before play and privately reveals $\vecpi_i$ to player $i$. The distribution $\rho$ is an \emph{$\epsilon$-conditional normal-form correlated equilibrium ($\epsilon$-conditional NFCE)}\footnote{At $\epsilon=0$, this is the usual normal-form correlated-equilibrium condition, written conditional on each recommendation with positive probability. For $\epsilon>0$, it is stronger than the usual (\emph{ex-ante}) condition that only bounds the expected benefit of a deviation map.} if, for every player $i$, every recommended policy $\vecpi_i$ with $\rho_i(\vecpi_i)>0$, and every possibly randomized, history-dependent deviation policy $\vecpi'_i = \vecpi_i'(\vecpi_i)$,
% \begin{equation}\label{eq:nfce}
%     \mathbb{E}\bigl[u_i(\vecpi_i,\vecpi_{-i})\mid \vecpi_i\bigr]
%     \ge
%     \mathbb{E}\bigl[u_i(\vecpi'_i,\vecpi_{-i})\mid \vecpi_i\bigr]
%     -\epsilon.
% \end{equation}
% \end{definition}
% Our hardness results hold even under the weaker \emph{ex-ante} notion, which we call an $\epsilon$-NFCE.
\section{Computing constant-expectation correlated equilibrium}
\label{sec:cece}

We first focus on approximately computing \emph{constant-expectation correlated equilibria}---a refinement of the usual notion---in normal-form games. This refinement of correlated equilibrium will be key to our results about Markov games.

\begin{definition}[Constant-expectation CE]\label[definition]{def:constant-expectation}
A distribution $\mu$ is an \emph{$(\epsilon,\delta)$-constant-expectation correlated equilibrium} if it is a $\delta$-CE and, for every player $i \in [n]$, there exists $v_i\in\mathbb{R}$ such that
\[
    \sum_{a_i\in\supp(p_i^\mu)}p_i^\mu(a_i)
    \left|\mathbb{E}_{\veca\sim\mu}
    [u_i(\veca)\mid a_i]-v_i\right|
    \le\epsilon.
\]
\end{definition}
Equivalently, $\sum_{a_i\in\Acal_i}|w_i^\mu(a_i)-v_i p_i^\mu(a_i)|\le\epsilon$. Our algorithm will in fact satisfy exact obedience and the stronger condition $\ell_i \le \mathbb{E}_{\veca\sim\mu} [u_i(\veca)\mid a_i] \le \ell_i+\epsilon$ for every supported $a_i$, whereas our hardness results hold for the more permissive averaged condition.

To clarify this concept, we next provide an illustrative example.

\begin{example}
    We consider the following bimatrix game \(\mathcal{G}\), where \(\mathcal A_1 = \{\mathsf{U},\mathsf{D}\}\) and
\(\mathcal A_2=\{\mathsf{L},\mathsf{R}\}\):
\[
\renewcommand{\arraystretch}{1.25}
\setlength{\arraycolsep}{12pt}
\begin{array}{c|cc}
     & \mathsf{L}     & \mathsf{R} \\ \hline
\mathsf{U}    & \cellcolor{payoffgray}(1,1) & (0,0) \\
\mathsf{D}    & (0,0) & \cellcolor{payoffgray}(1,1)
\end{array}
\]
Let $\mu$ be the uniform distribution over the pure profiles $(\mathsf{U}, \mathsf{L})$ and $( \mathsf{D}, \mathsf{R})$. This is not a Nash equilibrium, because $\mu$ is not a product distribution. It is, however, a constant-expectation CE. On the other hand, the same distribution is not a constant-expectation CE in the following bimatrix game.
\[
\renewcommand{\arraystretch}{1.25}
\setlength{\arraycolsep}{12pt}
\begin{array}{c|cc}
     & \mathsf{L}     & \mathsf{R} \\ \hline
\mathsf{U}    & \cellcolor{payoffgray}(1,1) & (0,0) \\
\mathsf{D}    & (0,0) & \cellcolor{payoffgray}(\nicefrac{1}{2},\nicefrac{1}{2})
\end{array}
\]
Nevertheless, $\mu$ is a correlated equilibrium in both games.
\end{example}

\paragraph{Linear programming-based algorithm.} The basic idea is to exhaustively search over all possible values in~\Cref{def:constant-expectation}, up to some small granularity. In particular, for a rational $\epsilon>0$, we let $G_\epsilon \defeq\{k\epsilon:k=0,1,\ldots,\lfloor 1/\epsilon\rfloor\}$ be a discretization of the $[0, 1]$ interval. For a fixed $\vecell=(\ell_i)_{i=1}^n\in G_\epsilon^n$, we consider the following linear feasibility program in $\mu$:

\begin{equation}\label{eq:grid-lp}
\begin{aligned}
    \mu
    &\in \operatorname{CE}(\mathcal{G}),\\
    \ell_i p_i^\mu(a_i)
    &\le w_i^\mu(a_i) \le (\ell_i+\epsilon)p_i^\mu(a_i)
    &&\text{for every $i \in [n]$ and $a_i\in \Acal_i$}.
\end{aligned}
\end{equation}
Above, $\operatorname{CE}(\mathcal{G})$ denotes the set of correlated equilibria of $\mathcal{G}$. We use the notation $A \defeq \max_{1 \leq i \leq n} |\mathcal{A}_i|$. 

\begin{restatable}{proposition}{gridlp}
\label[proposition]{prop:grid-lp}
There exists $\vecell \in G_\epsilon^n$ for which the linear program~\eqref{eq:grid-lp} is feasible, and every feasible solution is an $(\epsilon/2,0)$-constant-expectation correlated equilibrium. Moreover, $\ell_i\le\mathbb{E}_{\veca\sim\mu}[u_i(\veca)\mid a_i]\le\ell_i+\epsilon$ for every player $i$ and supported action $a_i$. Therefore, if $L$ is an upper bound on the bit complexity of the payoff tensor of $\mathcal{G}$, the running time is $(A/\epsilon)^{O(n)}\poly(L)$.
\end{restatable}

The running time of this algorithm scales exponentially with the number of players. In~\Cref{sec:PPAD}, we will show that this obstacle is fundamental.

\section{Efficient backward induction for NFCE in Markov games}
\label{sec:Markov}

Let $\Scal\defeq\{(t,s):t\in[H],\ s\in\Scal_t\}$ and $S\defeq|\Scal| = \sum_t|\Scal_t|$. We initialize $\ell_{i,H+1}(s)\defeq0$ for every player $i$ and terminal state $s\in\Scal_{H+1}$. We then process stages in reverse order. At a state $s\in\Scal_t$, we define an auxiliary normal-form game with payoff
\begin{equation}\label{eq:aux-payoff}
    g_{i,t,s}(\veca)
    \defeq r_{i,t}(s,\veca)
      +\sum_{s'\in\Scal_{t+1}}\mathbb{P}_t(s'\mid s,\veca)\ell_{i,t+1}(s').
\end{equation}
Using~\Cref{prop:grid-lp}, we can find a correlated equilibrium
$\mu_{t,s}\in\DeltaSet(\Acal_{t,s})$ and values $\ell_{i,t}(s)$ satisfying
\begin{equation}\label{eq:state-constant-expectation}
    \ell_{i,t}(s)
    \le
    \mathbb{E}_{\veca\sim\mu_{t,s}}
       \bigl[g_{i,t,s}(\veca)\mid a_i\bigr]
    \le
    \ell_{i,t}(s)+\epsilon
\end{equation}
for every $i$ and every $a_i\in \Acal_{i,t,s}$ with positive marginal probability under $\mu_{t,s}$. After all states have been processed, we independently sample
\begin{equation}\label{eq:local-sampling}
    \veca_{t,s}=(a_{1,t,s},\ldots,a_{n,t,s})\sim\mu_{t,s}
\end{equation}
for every $(t,s)\in\Scal$. The sampled pure Markov policy of player $i$ is
\begin{equation}\label{eq:sampled-policy}
    \pi_{i,t}(s) \defeq a_{i,t,s}.
\end{equation}
We define $\rho$ to be the resulting distribution over pure Markov policy profiles, succinctly represented through $(\mu_{t,s})_{(t,s)\in\Scal}$.

\begin{restatable}{lemma}{normaliz}
\label[lemma]{lem:well-defined}
For every $t,s,i$, $0\le \ell_{i,t}(s)\le (H-t+1)/H$, and the auxiliary payoffs at stage $t$ lie in $[0,(H-t+1)/H]$.
\end{restatable}

This shows that the payoffs lie in $[0,1]$, so we can apply~\Cref{prop:grid-lp}. Now, by virtue of independence, the constructed distribution satisfies the following key property.

\begin{restatable}{lemma}{factoriz}
\label[lemma]{lem:factorization}
Consider a player $i$ and a pure Markov policy $\vecpi_i$ with $\rho_i(\vecpi_i)>0$. Conditional on $\vecpi_i$, the random vectors
$\{\vecpi_{-i,t}(s):(t,s)\in\Scal\}$ are independent. At $(t,s)$, their \emph{conditional law} is
\begin{equation}\label{eq:conditional-local-law}
    \nu_{i,t,s}^{\vecpi_i}(\veca_{-i})
    \defeq\mathbb{P}[\vecpi_{-i,t}(s)=\veca_{-i}\mid \vecpi_i]
    =\frac{\mu_{t,s}(\pi_{i,t}(s),\veca_{-i})}
    {\mu_{i,t,s}(\pi_{i,t}(s))},
\end{equation}
where $\mu_{i,t,s}$ is the $i$th marginal of $\mu_{t,s}$.
\end{restatable}
This will allow us to argue about the deviation benefit of a player $i$ upon receiving its recommendation. In particular, conditioning on $\vecpi_i$, the other players' policy components are independent across states. Conditioning further on the history retains the law at $(t, s)$ as $\nu_{i,t,s}^{\vecpi_i}$ per~\eqref{eq:conditional-local-law}.

\begin{corollary}
\label[corollary]{cor:beliefs}
For a player $i$, consider any possibly randomized, history-dependent deviation conditioned on $\vecpi_i$. At any history $h_t$ at state $s$ having positive probability under this deviation, the conditional distribution of the opponents' prescribed actions at $(t,s)$ is $\nu_{i,t,s}^{\vecpi_i}$.
\end{corollary}
\subsection{The induced Markov decision process}

We now leverage this notion of constant-expectation correlated equilibrium for Markov games.

For a player $i$ and $\vecpi_i\in\supp(\rho_i)$, we will analyze the \emph{conditional MDP}
$\Mcal_i(\vecpi_i)$ with state spaces $\Scal_t$, action sets $\Acal_{i,t,s}$, expected rewards
\begin{equation}
\label{eq:mdp-reward}
    \bar r_{i,t}^{\vecpi_i}(s,a_i')
    \defeq\sum_{\veca_{-i}}
      \nu_{i,t,s}^{\vecpi_i}(\veca_{-i})
      r_{i,t}(s,a_i',\veca_{-i}),
\end{equation}
and transition kernels
\begin{equation}
\label{eq:mdp-transition}
    \overline{\mathbb{P}}_{i,t}^{\vecpi_i}(s'\mid s,a_i')
    \defeq\sum_{\veca_{-i}}
      \nu_{i,t,s}^{\vecpi_i}(\veca_{-i})
      \mathbb{P}_t(s'\mid s,a_i',\veca_{-i}).
\end{equation}

\begin{restatable}[Reduction to an MDP]{lemma}{mdpreduction}
\label[lemma]{lem:mdp-reduction}
Conditional on $\vecpi_i$, the expected payoff of every possibly randomized, history-dependent deviation of player $i$ in the Markov game equals its expected payoff in the MDP $\Mcal_i(\vecpi_i)$.
\end{restatable}

This in particular implies that an optimal deviation may be chosen as deterministic and Markovian.

We define the Bellman values of this MDP by $V_{i,H+1}^{*,\vecpi_i}(s) \defeq 0$,
\begin{align*}
    Q_{i,t}^{*,\vecpi_i}(s,a_i')
    \defeq\bar r_{i,t}^{\vecpi_i}(s,a_i')
      +\sum_{s'}\overline{\mathbb{P}}_{i,t}^{\vecpi_i}(s'\mid s,a_i')
        V_{i,t+1}^{*,\vecpi_i}(s'),
\end{align*}
and $V_{i,t}^{*,\vecpi_i}(s) \defeq\max_{a_i'\in \Acal_{i,t,s}}
      Q_{i,t}^{*,\vecpi_i}(s,a_i')$. Now, let
$V_{i,t}^{\vecpi_i}(s)$ denote player $i$'s value in $\Mcal_i(\vecpi_i)$ when following the recommended Markov policy $\vecpi_i$ from stage $t$ onward. We have
\begin{equation}\label{eq:obedient-bellman}
\begin{split}
    V_{i,t}^{\vecpi_i}(s)
    ={}&\sum_{\veca_{-i}}\nu_{i,t,s}^{\vecpi_i}(\veca_{-i})
    \Bigl(
      r_{i,t}(s,\pi_{i,t}(s),\veca_{-i})\\
      &\hspace{35mm}
      +\sum_{s'}\mathbb{P}_t(s'\mid s,\pi_{i,t}(s),\veca_{-i})
        V_{i,t+1}^{\vecpi_i}(s')
    \Bigr).
\end{split}
\end{equation}

Having computed a constant-expectation CE in each auxiliary game, we can derive the following.

\begin{restatable}{lemma}{onestep}
\label[lemma]{lem:one-step}
Let $a_i \defeq\pi_{i,t}(s)$. For every action
$a_i'\in \Acal_{i,t,s}$,
\begin{equation*}
    \label{eq:obedient-surrogate-lower}
    \mathbb{E}_{\veca_{-i}\sim\nu_{i,t,s}^{\vecpi_i}}
      [g_{i,t,s}(a_i',\veca_{-i})]
    \le \ell_{i,t}(s)+\epsilon \text{ and } \mathbb{E}_{\veca_{-i}\sim\nu_{i,t,s}^{\vecpi_i}}
      [g_{i,t,s}(a_i,\veca_{-i})]
    \ge \ell_{i,t}(s).
\end{equation*}
\end{restatable}

\begin{restatable}{lemma}{valuebounds}
\label[lemma]{lem:value-bounds}
For every $i$, every $\vecpi_i\in\supp(\rho_i)$, every stage $t$, and every state $s\in\Scal_t$,
\begin{equation*}
    V_{i,t}^{\vecpi_i}(s) \geq \ell_{i,t}(s) \text{ and } V_{i,t}^{*,\vecpi_i}(s)
    \le \ell_{i,t}(s)+(H-t+1) \epsilon.
\end{equation*}
\end{restatable}

\subsection{Main result}

\begin{algorithm}[t]
\caption{Backward induction for NFCE in Markov games}
\label{alg:nfce}
\DontPrintSemicolon
\SetKw{Break}{break}
\SetKw{Return}{return}
\KwIn{Finite-horizon Markov game satisfying
    \Cref{ass:normalization}; accuracy $\epsilon\in(0,1]$.}
\KwOut{An $\epsilon$-conditional NFCE represented by
    $(\mu_{t,s})_{(t,s)\in\Scal}$, sampled independently.}

$\eta\gets\epsilon/H$\;
$\ell_{i,H+1}(s)\gets0$ for every $i\in[n]$
    and $s\in\Scal_{H+1}$\;
\For{$t=H,H-1,\ldots,1$}{
    \ForEach{$s\in\Scal_t$}{
        Define, for every $i\in[n]$ and $\veca\in\Acal_{t,s}$,
        \[
            g_{i,t,s}(\veca)
            \gets r_{i,t}(s,\veca)
            +\sum_{s'\in\Scal_{t+1}}
            \mathbb{P}_t(s'\mid s,\veca)\ell_{i,t+1}(s').
        \]
        \ForEach{$\vecell\in G_\eta^n$}{
            Solve~\eqref{eq:grid-lp} with payoffs
                $(g_{i,t,s})_{i=1}^n$, vector $\vecell$,
                and width $\eta$\;
            \If{feasible, with solution $\mu$}{
                $\mu_{t,s}\gets\mu$ and
                    $\ell_{i,t}(s)\gets\ell_i$
                    for every $i\in[n]$\;
                \Break\;
            }
        }
    }
}
\Return{$(\mu_{t,s})_{(t,s)\in\Scal}$, inducing $\rho$
    via~\eqref{eq:local-sampling}--\eqref{eq:sampled-policy}}\;
\end{algorithm}

We are now ready to establish our main result.

\begin{theorem}[Backward-induction conditional NFCE guarantee]\label{thm:main}
The distribution $\rho$ over pure Markov policies produced by~\Cref{alg:nfce} is an $\epsilon$-conditional NFCE. Moreover, assuming the stage games and transitions are given explicitly with bit complexity at most $L$, and $A\defeq\max_{i,t,s}|\Acal_{i,t,s}|$, it runs in time
\begin{equation*}
    S\left(\frac{AH}{\epsilon}\right)^{O(n)}\poly(L).
\end{equation*}
\end{theorem}

\begin{proof}
We consider a player $i$ and a recommendation $\vecpi_i\in\supp(\rho_i)$. By~\Cref{lem:mdp-reduction}, the optimal conditional deviation value from the initial state distribution $\lambda$ is $\mathbb{E}_{s_1\sim\lambda}[V_{i,1}^{*,\vecpi_i}(s_1)]$. Following the recommendation yields $\mathbb{E}_{s_1\sim \lambda}[V_{i,1}^{\vecpi_i}(s_1)]$. Thus, \Cref{lem:value-bounds} implies
\begin{align*}
\mathbb{E}_{s_1\sim\lambda}
  [V_{i,1}^{*,\vecpi_i}(s_1)-V_{i,1}^{\vecpi_i}(s_1)] \le H \eta = \epsilon.
\end{align*}
The output is represented by the $S$ distributions $\mu_{t,s}$, which are to be sampled independently.
\end{proof}
\section{Hardness results}
\label{sec:PPAD}

We now turn to our hardness results. We begin by establishing \PPAD-hardness for constant-expectation CEs, and then translate those results to NFCEs in Markov games.

\subsection{Succinct multi-player games}
\label{sec:polymatrix}

We first examine succinct many-player games. Here, we require an output representation from which the marginal probabilities can be computed in polynomial time, such as a mixture of product distributions~\citep{Papadimitriou08:Computing}.

Our hardness result holds even in \emph{polymatrix games}. A \emph{polymatrix game} is specified by a graph $([n],E)$ and a payoff matrix $\matM_{ii'}$ for every edge $(i,i') \in E$. Let $N_i$ be the \emph{neighbor set} of player $i$. At a pure action profile $\veca$, player $i$ receives $u_i(\veca)=\sum_{i'\in N_i} \matM_{ii'}(a_i,a_{i'})$.  We say that the game is normalized if $u_i (\veca)\in[0,1]$ for every player and action profile. Our reduction will use the hardness result of~\citet{Deligkas24:PureCircuit}, which builds on~\citet{Rubinstein15:Inapproximability}.

\begin{theorem}[\citealp{Deligkas24:PureCircuit}]
\label{thm:ne-hardness}
For every constant $0 < \delta<1/12$, computing a $\delta$-Nash equilibrium in a normalized polymatrix game of maximum degree three with two actions per player is \PPAD-hard.
\end{theorem}
Let $\mathcal{G}$ be such a polymatrix game with $\Acal_i=\{0,1\}$ for every player. We may assume that the source game has no isolated players. We define the payoff difference
\[
    d_i(\veca_{-i})
    \defeq u_i(1,\veca_{-i})-u_i(0,\veca_{-i})
    =\sum_{i'\in N_i}
      \bigl(\matM_{ii'}(1,a_{i'})-\matM_{ii'}(0,a_{i'})\bigr).
\]
Normalization implies $d_i(\veca_{-i})\in[-1,1]$. For a constant $0<\epsilon<1/24$, we set $\alpha\defeq1/4+6\epsilon$, so $12\epsilon<\alpha<1/2$. We construct a transformed game $\mathcal{G}'$ on the same graph and action sets, with payoffs
\[
    u_i'(0,\veca_{-i})
    \defeq\frac12+\alpha d_i(\veca_{-i}) \text{ and }
    u_i'(1,\veca_{-i})
    \defeq\frac12+\frac12d_i(\veca_{-i}).
\]

This game is still a polymatrix game, as we point out below.

\begin{restatable}{lemma}{transformedpolymatrix}\label[lemma]{lem:transformed-polymatrix}
The game $\mathcal{G}'$ is a normalized polymatrix game of maximum degree three with two actions per player.
\end{restatable}

For polymatrix games, the utilities can be expressed in terms of the marginal strategies.

\begin{restatable}{lemma}{marginalpayoffdifference}
\label[lemma]{lem:marginal-payoff-difference}
Let $\mu\in\DeltaSet(\Acal)$, and let $\vecx$ be its marginal strategy profile.  Then, for every player $i$,
\[
    \mathbb{E}_{\veca\sim\mu}[d_i(\veca_{-i})]
    =u_i(1,\vecx_{-i})-u_i(0,\vecx_{-i}).
\]
\end{restatable}

The following lemma now makes crucial use of the constant expectation property.

\begin{restatable}{lemma}{conditionalpayoffdifference}\label[lemma]{lem:conditional-payoff-difference}
Let $\mu$ be an $(\epsilon,\delta)$-constant-expectation CE of $\mathcal{G}'$. For any player $i$,
\[
    \left(\frac12-\alpha\right)
    \sum_{a_i\in\supp(p_i^\mu)}p_i^\mu(a_i)
    \max\left\{(1-2a_i)\mathbb{E}_\mu[d_i(\veca_{-i})\mid a_i],0\right\}
    \le\delta.
\]
If both actions have positive marginal probability, then
\[
    \min\{p_i^\mu(0),p_i^\mu(1)\}
    \left|\frac12\mathbb{E}_\mu[d_i(\veca_{-i})\mid a_i=1]
      -\alpha\mathbb{E}_\mu[d_i(\veca_{-i})\mid a_i=0]\right|
    \le\epsilon.
\]
\end{restatable}

This lemma, together with simple algebra, implies a certain equilibrium collapse: the marginals of any constant-expectation CE reveal a Nash equilibrium of the original game.

\begin{restatable}{lemma}{marginalsarene}\label[lemma]{lem:marginals-are-ne}
If $\mu$ is an $(\epsilon,\delta)$-constant-expectation correlated equilibrium of $\mathcal{G}'$ and $\vecx$ is its marginal strategy profile, then $\vecx$ is an $\eta$-Nash equilibrium of $\mathcal{G}$, where $\eta = \epsilon/\alpha+\delta/(\alpha(1-2\alpha))$.
\end{restatable}

Combining with the hardness result of~\citet{Deligkas24:PureCircuit} (\Cref{thm:ne-hardness}), we arrive at the following.

\begin{restatable}[\PPAD-hardness]{theorem}{constantexpectationhardness}\label{thm:constant-expectation-hardness}
For every pair of constants such that $0 < \epsilon < 1/24$ and $0 \leq \delta < ( 1 - 24 \epsilon)^2/96$, computing an $(\epsilon,\delta)$-constant-expectation correlated equilibrium in a normalized polymatrix game of maximum degree three with two actions per player is \PPAD-hard.
\end{restatable}
%In particular, this holds for $\epsilon=\delta=1/192$.

\subsection{Two-player games}
\label{sec:two-player-hardness}

We now turn our attention to two-player (bimatrix) games. We reduce from the problem of computing an exact Nash equilibrium of a rank-$2$ bimatrix game, which is \PPAD-hard~\citep{Chen:Rank2}. Let $\mathcal{G}$ be such a game, so the sum of its payoff matrices has rank at most two. We write
\begin{equation}
\label{eq:rank-two-decomposition}
    u_1(\veca)+u_2(\veca)
    =f_1(a_1)g_1(a_2)+f_2(a_1)g_2(a_2),
\end{equation}
where $f_1,f_2:\Acal_1\to[-1,1]$ and
$g_1,g_2:\Acal_2\to[-1,1]$. For a parameter $K\ge 3$, we consider the game $\mathcal{G}'$ on the same action sets with payoffs
\begin{align}
    u_1'(\veca)
    &\defeq \frac{u_1(\veca)+2g_2(a_2)+K(g_1(a_2)+2)}{4K},
    \label{eq:rank-two-transformation}\\
    u_2'(\veca)
    &\defeq \frac{u_2(\veca)+2f_1(a_1)+K(f_2(a_1)+2)}{4K}.
    \nonumber
\end{align}
As we point out below, $\mathcal{G}'$ is strategically equivalent to $\mathcal{G}$.

\begin{restatable}{lemma}{ranktwotransformation}\label[lemma]{lem:rank-two-transformation}
$\mathcal{G}'$ has payoffs in $[0,1]$, has rank at most two, and has the same set of CEs as $\mathcal{G}$.
Moreover, $\mu$ is a $\delta$-CE of $\mathcal{G}'$ if and only if it is a $(4K\delta)$-CE of $\mathcal{G}$.
\end{restatable}

The role of constant expectation is to make recommendations nearly uninformative about opponent-dependent terms, on average. For Player~1, when $K \gg 1$, the term $Kg_1(a_2)$ dominates the remaining payoff. So, we can control the expected absolute deviation of the conditional expectation of $g_1(a_2)$ from its unconditional expectation. Player~2 similarly controls $f_2(a_1)$.
These are precisely the two correlations that can separate a correlated equilibrium from a Nash equilibrium, as we establish in~\Cref{lem:rank-two-amplification}.
In what follows, we use the notation $\operatorname{Cov}_\mu(X,Y)\defeq
\mathbb{E}_\mu[XY] - \mathbb{E}_\mu[X]\mathbb{E}_\mu[Y]$.

\begin{restatable}{lemma}{ranktwocovariance}\label[lemma]{lem:rank-two-covariance}
If $\mu$ is an $\eta$-CE of $\mathcal{G}$ with marginal
strategy profile $\vecx$, then
\[
    \sum_{i=1}^2\left(
        \max_{a_i' \in \mathcal{A}_i}u_i(a_i',\vecx_{-i})-u_i(\vecx)
    \right)
    \le 2\eta+\sum_{i=1}^2
    \operatorname{Cov}_\mu(f_i(a_1),g_i(a_2)).
\]
\end{restatable}

As discussed, the averaged constant-expectation property suffices to control the covariance terms.

\begin{restatable}{lemma}{ranktwoamplification}\label[lemma]{lem:rank-two-amplification}
If $\mu$ is an $(\epsilon,\delta)$-constant-expectation correlated equilibrium of
$\mathcal{G}'$, then, for $i\in\{1,2\}$,
\[
    \left|\operatorname{Cov}_\mu(f_i(a_1),g_i(a_2))\right|
    \le 8\epsilon+\frac{6}{K}.
\]
\end{restatable}

Combining, we arrive at the following consequence, where we assume that $\epsilon$ and $\delta$ are given as part of the input in binary, so that $\log(1/\epsilon)$ and $\log(1/\delta)$ are polynomial.

\begin{restatable}{theorem}{twoplayerconstantexpectationhardness}\label[theorem]{thm:two-player-constant-expectation-hardness}
Computing an $(\epsilon,\delta)$-constant-expectation correlated equilibrium is \PPAD-hard in bimatrix games when $\epsilon$ and $\delta$ are exponentially small.
\end{restatable}

\subsection{Hardness for NFCE in Markov games}

We finally circle back to Markov games. Our main hardness result establishes that relying on constant-expectation CEs in backward induction is fundamental:

\begin{restatable}{theorem}{nfceTransfer}
\label{thm:nfce-transfer}
    Let $\cG$ be an $n$-player normal-form game with utilities bounded in $[0,1]$, and $0\le\eps\le 1/(16 n)$. There is a Markov game $\cM$ with $n+1$ players, horizon $H=2$, $S=n+1$ states, and stage rewards bounded in $[0, 1/2]$, such that every $\eps$-NFCE with independent state recommendations yields a $(16n\eps,16\eps)$-constant-expectation CE of $\cG$.
\end{restatable}

By independent state recommendations, we mean that the joint action recommendations are sampled independently across pairs $(t, s)$; recommendations to different players at the same state may be correlated. Combining the above result with \Cref{thm:constant-expectation-hardness} and \Cref{thm:two-player-constant-expectation-hardness}, we therefore immediately derive \PPAD-hardness of computing NFCE with independent state recommendations.
\begin{corollary}
\label{cor:nfce-independent-hardness}
Computing an $\eps$-NFCE with independent state recommendations is \PPAD-hard for Markov games satisfying~\Cref{ass:normalization}, even with horizon $H=2$, in each of the following cases:
\begin{enumerate}
    \item $(n+1)$-player succinct games, two actions per player at each state, and $\eps=1/(3072n)$.
    \item Three-player games, a constant number of states, and exponentially small $\eps$.
\end{enumerate}
\end{corollary}
\section{Conclusion}

We established an efficient algorithm for computing conditional normal-form correlated equilibria in finite-horizon Markov games with a fixed number of players. Our key insight is that constant-expectation correlated equilibria, combined with independent sampling across states, make backward induction sound even when players observe their entire recommended policies before play. Our results suggest several interesting questions for future research. First, our algorithm is neither decentralized nor model free. Addressing this limitation is an important challenge for future work, but we believe our work lays the groundwork for such extensions. It would also be interesting to understand whether there are efficient algorithms when the number of players is part of the input. Our \PPAD-hardness rules this out under state independent recommendations, but perhaps one can circumvent that result by introducing correlation between states. Finally, the complexity of computing NFCEs in other game representations, most notably the extensive form, remains a challenging open question.

\section*{AI Use Statement}

The initial idea for this paper, and the main idea for the proof of \Cref{theorem:informal}, are due to the authors. Various versions of GPT were used to assist in proofs of some of the lower bounds. All proofs were checked, understood, and rewritten by the authors.

\section*{Acknowledgments}

I.A. and T.S. are supported by NIH award A240108S001, the Vannevar Bush Faculty Fellowship ONR N00014-23-1-2876, and National Science Foundation grant RI-2312342. We are grateful to Xi Chen and Mihalis Yannakakis for sharing their manuscript~\citep{Chen:Rank2}. I.A. additionally thanks Ioannis Panageas for helpful discussions concerning constant-rank games. C.D. was supported by a Simons Investigator Award, a Simons Collaboration on Algorithmic Fairness, ONR MURI grant N00014-25-1-2116, ONR grant N00014-25-1-2296. G.F. was supported in part by the National Science Foundation award CCF-2443068, the Office of Naval Research grant N000142512296, and an AI2050 Early Career Fellowship.

\bibliographystyle{plainnat}
\bibliography{refs}

\clearpage
\appendix

\section{Omitted proofs}

This section contains the omitted proofs. 

\subsection{Proofs from Section~\ref{sec:cece}}

We begin with~\Cref{prop:grid-lp}.

\gridlp*

\begin{proof}
By Nash's theorem~\citep{Nash51:Noncooperative}, $\mathcal{G}$ admits a Nash equilibrium $\vecx=(\vecx_i)_{i=1}^n$. Let $\mu=\bigotimes_i \vecx_i$ be the induced product distribution, which is clearly a correlated equilibrium. Every action $a_i\in\supp(\vecx_i)$ is a best response to $\vecx_{-i}$, so $\mathbb{E}_{\veca\sim\mu}[u_i(\veca)\mid a_i] = u_i(a_i,\vecx_{-i}) $.  For $a_i\notin\supp(\vecx_i)$, both $p_i^\mu(a_i)$ and $w_i^\mu(a_i)$ are zero. For $a_i \in \supp(\vecx_i)$, let $\ell_i\defeq \epsilon \lfloor u_i(a_i,\vecx_{-i}) /\epsilon\rfloor$.
Since $u_i(a_i,\vecx_{-i}) \in[0,1]$, we have $\ell_i\in G_\epsilon$ and
$\ell_i\le u_i(a_i,\vecx_{-i}) < \ell_i+\epsilon$. This shows that $\mu$ must be a feasible point of the LP induced by $\vecell$.

Conversely, every feasible solution to~\eqref{eq:grid-lp} is a correlated equilibrium. Whenever $p_i^\mu(a_i)>0$, dividing the last line by $p_i^\mu(a_i)$ gives
$\ell_i\le \mathbb{E}_{\veca\sim\mu}[u_i(\veca)\mid a_i]\le \ell_i+\epsilon$. Thus, choosing $v_i=\ell_i+\epsilon/2$ bounds the absolute deviation by $\epsilon/2$ for every supported recommendation, and hence also after averaging over recommendations, giving an $(\epsilon/2,0)$-constant-expectation correlated equilibrium as required by~\Cref{def:constant-expectation}.

Finally, the claimed running time follows because $G_\epsilon^n$ contains $(1/\epsilon)^{O(n)}$ points, and each induced linear program can be solved in time polynomial in $A^n(L + \log(1/\epsilon))$.
\end{proof}

\subsection{Proofs from Section~\ref{sec:Markov}}

We continue with the proofs from~\Cref{sec:Markov}.

\normaliz*

\begin{proof}
At stage $H+1$, the claim holds because the continuation values are zero. Suppose it holds at stage $t+1$. By the assumed normalization (\Cref{ass:normalization}) and the definition of the auxiliary normal-form games, we have
\[
    0\le g_{i,t,s}(\veca)
    \le \frac1H+\frac{H-t}{H}
    =\frac{H-t+1}{H}.
\]
Since \(\ell_{i,t}(s)\ge0\) belongs to the grid and is bounded above by a supported action’s conditional expected auxiliary payoff, the same upper bound holds for \(\ell_{i,t}(s)\).
\end{proof}

\factoriz*

\begin{proof}
Let $\rho_i(\vecpi_i)>0$. Independence in the constructed distribution~\eqref{eq:local-sampling} gives 
\begin{equation*}
\rho_i(\vecpi_i)
    = \prod_{(t,s)\in\Scal} \mu_{i,t,s}(\pi_{i,t}(s)).    
\end{equation*}
Moreover, for any
$(\veca_{-i,t,s})_{(t,s)\in\Scal}$,
\begin{align*}
\mathbb{P}\bigl[
    \vecpi_{-i,t}(s)=\veca_{-i,t,s}\text{ for all $(t,s)$}
    \mid \vecpi_i
  \bigr] =
\frac{
  \prod_{(t,s)\in\Scal}
  \mu_{t,s}(\pi_{i,t}(s),\veca_{-i,t,s})
}{
  \prod_{(t,s)\in\Scal}
  \mu_{i,t,s}(\pi_{i,t}(s))
}
=
\prod_{(t,s)\in\Scal}
\nu_{i,t,s}^{\vecpi_i}(\veca_{-i,t,s}),
\end{align*}
as claimed.
\end{proof}

\mdpreduction*

\begin{proof}
\Cref{cor:beliefs} shows that after every history at $(t,s)$, and for every deviating action $a_i'$, the conditional expected immediate reward and transition probabilities are given by~\eqref{eq:mdp-reward} and~\eqref{eq:mdp-transition}, respectively. Applying iterated expectation over the stages, the expected cumulative reward in the Markov game conditional on $\vecpi_i$ is equal to that in $\Mcal_i(\vecpi_i)$.
\end{proof}

\onestep*

\begin{proof}
Because $\vecpi_i\in\supp(\rho_i)$, the recommendation $a_i$ has positive marginal probability under $\mu_{t,s}$. Thus,
\[
    \mathbb{E}_{\veca_{-i}\sim\nu_{i,t,s}^{\vecpi_i}}[g_{i,t,s}(a_i,\veca_{-i})]
    \ge
    \mathbb{E}_{\veca_{-i}\sim\nu_{i,t,s}^{\vecpi_i}}[g_{i,t,s}(a_i',\veca_{-i})].
\]
Combining this with~\eqref{eq:state-constant-expectation} gives the claim.
\end{proof}

\valuebounds*

\begin{proof}
We proceed by backward induction, starting from the lower bound. At $t=H+1$, the claim is immediate. Suppose that it holds at stage $t+1$, and let $a_i\defeq\pi_{i,t}(s)$. Using~\eqref{eq:obedient-bellman}, the inductive hypothesis
$V_{i,t+1}^{\vecpi_i}(s')\ge\ell_{i,t+1}(s')$, and~\Cref{lem:one-step}, we obtain
\begin{align*}
    V_{i,t}^{\vecpi_i}(s)
    &\ge
    \mathbb{E}_{\veca_{-i}\sim\nu_{i,t,s}^{\vecpi_i}}
    \left[
      r_{i,t}(s,a_i,\veca_{-i})
      +\sum_{s'}\mathbb{P}_t(s'\mid s,a_i,\veca_{-i})
       \ell_{i,t+1}(s')
    \right]\\
    &=\mathbb{E}_{\veca_{-i}\sim\nu_{i,t,s}^{\vecpi_i}}
      [g_{i,t,s}(a_i,\veca_{-i})]
    \ge \ell_{i,t}(s).
\end{align*}

For the upper bound, the Bellman recursion, the inductive hypothesis
$V_{i,t+1}^{*,\vecpi_i}(s') \le \ell_{i,t+1}(s')+(H-t)\epsilon$, and~\Cref{lem:one-step} give
\begin{align*}
    V_{i,t}^{*,\vecpi_i}(s)
    &\le
    \max_{a_i'\in \Acal_{i,t,s}}
    \mathbb{E}_{\veca_{-i}\sim\nu_{i,t,s}^{\vecpi_i}}
    \left[
      g_{i,t,s}(a_i',\veca_{-i})+(H-t)\epsilon
    \right]\\
    &\le \ell_{i,t}(s)+\epsilon+(H-t)\epsilon\\
    &=\ell_{i,t}(s)+(H-t+1)\epsilon.
\end{align*}
This completes the induction.
\end{proof}

\subsection{Proofs for Section~\ref{sec:PPAD}}

We conclude with the proofs from~\Cref{sec:PPAD}, starting from our hardness result for polymatrix games.

\transformedpolymatrix*

\begin{proof}
For each edge $(i,i')$, we define
\begin{align*}
    \matM_{ii'}'(0,a_{i'})
    &\defeq\frac{1}{2|N_i|}
      +\alpha\bigl(\matM_{ii'}(1,a_{i'})-\matM_{ii'}(0,a_{i'})\bigr),\\
    \matM_{ii'}'(1,a_{i'})
    &\defeq\frac{1}{2|N_i|}
      +\frac12\bigl(\matM_{ii'}(1,a_{i'})-\matM_{ii'}(0,a_{i'})\bigr).
\end{align*}
Summing over $i'\in N_i$ gives the specified payoffs. Because $d_i\in[-1,1]$ and $0<\alpha<1/2$, both payoffs lie in $[0,1]$. The interaction graph and action sets remain the same, and the construction has polynomial bit complexity.
\end{proof}

\marginalpayoffdifference*

\begin{proof}
Linearity and the polymatrix structure give
\begin{align*}
    \mathbb{E}_{\veca\sim\mu}[d_i(\veca_{-i})]
    &=\sum_{i'\in N_i}
      \mathbb{E}_{\veca\sim\mu}
      [\matM_{ii'}(1,a_{i'})-\matM_{ii'}(0,a_{i'})]\\
    &=\sum_{i'\in N_i}\sum_{a_{i'}\in\{0,1\}}x_{i'}(a_{i'})
      \bigl(\matM_{ii'}(1,a_{i'})-\matM_{ii'}(0,a_{i'})\bigr)\\
    &=u_i(1,\vecx_{-i})-u_i(0,\vecx_{-i}).
\end{align*}
\end{proof}

\conditionalpayoffdifference*

\begin{proof}
Since $1/2-\alpha>0$, we consider the deviation map that switches supported recommendation $0$ precisely when $\mathbb{E}_\mu[d_i(\veca_{-i})\mid a_i=0]>0$, switches supported recommendation $1$ precisely when $\mathbb{E}_\mu[d_i(\veca_{-i})\mid a_i=1]<0$, and otherwise obeys. Its expected gain is the sum of the positive gains from the two recommendations, so $\delta$-CE obedience gives
\[
    \left(\frac12-\alpha\right)
    \sum_{a_i\in\supp(p_i^\mu)}p_i^\mu(a_i)
    \max\left\{(1-2a_i)\mathbb{E}_\mu[d_i(\veca_{-i})\mid a_i],0\right\}
    \le\delta.
\]

Suppose now that $p_i^\mu(0), p_i^\mu(1)>0$. By~\Cref{def:constant-expectation}, there exists $v_i$ such that
\begin{align*}
    \epsilon
    &\ge p_i^\mu(0)\left|\mathbb{E}_\mu[u_i'(\veca)\mid a_i=0]-v_i\right|
      +p_i^\mu(1)\left|\mathbb{E}_\mu[u_i'(\veca)\mid a_i=1]-v_i\right|\\
    &\ge\min\{p_i^\mu(0),p_i^\mu(1)\}
      \left|\mathbb{E}_\mu[u_i'(\veca)\mid a_i=1]
        -\mathbb{E}_\mu[u_i'(\veca)\mid a_i=0]\right|\\
    &=\min\{p_i^\mu(0),p_i^\mu(1)\}
      \left|\frac12\mathbb{E}_\mu[d_i(\veca_{-i})\mid a_i=1]
        -\alpha\mathbb{E}_\mu[d_i(\veca_{-i})\mid a_i=0]\right|,
\end{align*}
which proves the claim.
\end{proof}

This lemma, together with simple calculations, establishes the claimed equilibrium collapse, as we formalize below.

\marginalsarene*

\begin{proof}
By~\Cref{lem:marginal-payoff-difference}, we have that for any player $i$,
\[
    u_i(1,\vecx_{-i})-u_i(0,\vecx_{-i})
    =\mathbb{E}_\mu[d_i(\veca_{-i})].
\]
If $p_i^\mu(0)=1$, obedience gives
\begin{align*}
    \max_{a_i' \in \{0, 1\} }u_i(a_i',\vecx_{-i})-u_i(\vecx) =\max\{u_i(1,\vecx_{-i})-u_i(0,\vecx_{-i}),0\}
      \le\frac{\delta}{1/2-\alpha}.
\end{align*}
Similarly, if $p_i^\mu(1)=1$,
\begin{align*}
    \max_{a_i' \in \{0, 1\}}u_i(a_i',\vecx_{-i})-u_i(\vecx) =\max\{u_i(0,\vecx_{-i})-u_i(1,\vecx_{-i}),0\}
      \le\frac{\delta}{1/2-\alpha}.
\end{align*}
Since $2\alpha<1$, these bounds are at most $\epsilon/\alpha+\delta/(\alpha(1-2\alpha))$. If both probabilities are positive,
\begin{align*}
    u_i(1,\vecx_{-i})-u_i(0,\vecx_{-i}) =p_i^\mu(0)\mathbb{E}_\mu[d_i(\veca_{-i})\mid a_i=0]
      +p_i^\mu(1)\mathbb{E}_\mu[d_i(\veca_{-i})\mid a_i=1].
\end{align*}
If $u_i(1,\vecx_{-i})\ge u_i(0,\vecx_{-i})$, then
\begin{align*}
    \max_{a_i' \in \{0, 1\}} u_i(a_i',\vecx_{-i})-u_i(\vecx) &=p_i^\mu(0)\bigl(u_i(1,\vecx_{-i})-u_i(0,\vecx_{-i})\bigr)\\
    &=2p_i^\mu(0)p_i^\mu(1)
      \left(\frac12\mathbb{E}_\mu[d_i(\veca_{-i})\mid a_i=1]
        -\alpha\mathbb{E}_\mu[d_i(\veca_{-i})\mid a_i=0]\right)\\
    &\qquad+p_i^\mu(0)\bigl(p_i^\mu(0)+2\alpha p_i^\mu(1)\bigr)
      \mathbb{E}_\mu[d_i(\veca_{-i})\mid a_i=0]\\
    &\quad\le2\epsilon+\frac{\delta}{1/2-\alpha},
\end{align*}
where we used the fact that $p_i^\mu(0)p_i^\mu(1)\le\min\{p_i^\mu(0),p_i^\mu(1)\}$, $p_i^\mu(0)+2\alpha p_i^\mu(1)\le1$, and we applied both bounds from~\Cref{lem:conditional-payoff-difference}.
Similarly, if $u_i(1,\vecx_{-i})\le u_i(0,\vecx_{-i})$,
\begin{align*}
    \max_{a_i' \in \{0, 1\}}u_i(a_i',\vecx_{-i})-u_i(\vecx) &= p_i^\mu(1)\bigl(u_i(0,\vecx_{-i})-u_i(1,\vecx_{-i})\bigr)\\
    &= \frac{p_i^\mu(0)p_i^\mu(1)}{\alpha}
      \left(\frac12\mathbb{E}_\mu[d_i(\veca_{-i})\mid a_i=1]
        -\alpha\mathbb{E}_\mu[d_i(\veca_{-i})\mid a_i=0]\right)\\
    &\qquad -p_i^\mu(1)\left(p_i^\mu(1)+\frac{p_i^\mu(0)}{2\alpha}\right)
      \mathbb{E}_\mu[d_i(\veca_{-i})\mid a_i=1]\\
    &\le\frac{\epsilon}{\alpha}
      +\frac{\delta}{\alpha(1-2\alpha)},
\end{align*}
where we used the fact that $p_i^\mu(1)+p_i^\mu(0)/(2\alpha)\le1/(2\alpha)$ and~\Cref{lem:conditional-payoff-difference}. Since $\alpha<1/2$, the second upper bound also dominates the first, proving the claim.
\end{proof}

We put everything together below.

\constantexpectationhardness*

\begin{proof}
We consider constants $0<\epsilon<1/24$ and $0\le\delta<(1-24\epsilon)^2/96$, and set $\alpha=1/4+6\epsilon$ and $\eta=\epsilon/\alpha+\delta/(\alpha(1-2\alpha))$. Since $\alpha(1-2\alpha)>0$, we have
\[
    \frac{\epsilon}{\alpha}
    +\frac{\delta}{\alpha(1-2\alpha)}<\frac1{12}
    \quad\Longleftrightarrow\quad
    \delta<(\alpha/12-\epsilon)(1-2\alpha)
    =\frac{(1-24\epsilon)^2}{96}.
\]
Given an instance $\mathcal{G}$ of the problem in~\Cref{thm:ne-hardness} for this $\eta<1/12$, construct $\mathcal{G}'$ as above. The construction is polynomial. From any $(\epsilon,\delta)$-constant-expectation correlated equilibrium of $\mathcal{G}'$, we compute its marginals $\vecx$. By~\Cref{lem:marginals-are-ne}, $\vecx$ is an $\eta$-Nash equilibrium of $\mathcal{G}$, and the result follows from~\Cref{thm:ne-hardness}.
\end{proof}

Next, we provide the proofs pertaining to the two-player reduction, which starts from a rank-2 game.

\ranktwotransformation*

\begin{proof}
Normalization is immediate. The rank bound follows from the fact that
\[
    4K(u_1'(\veca)+u_2'(\veca))
    =(f_1(a_1)+K)(g_1(a_2)+2)
      +(f_2(a_1)+2)(g_2(a_2)+K).
\]
Moreover, for every deviation map $\phi_i:\Acal_i\to\Acal_i$,
\[
    \mathbb{E}_\mu[u_i'(\phi_i(a_i),\veca_{-i})-u_i'(\veca)]
    =\frac1{4K}\mathbb{E}_\mu[u_i(\phi_i(a_i),\veca_{-i})-u_i(\veca)],
\]
Thus, $\mu$ is a $\delta$-CE of $\mathcal{G}'$ if and only if it is a $(4K\delta)$-CE of $\mathcal{G}$.
\end{proof}

\ranktwocovariance*

\begin{proof}
For each player $i$, we have $\mathbb{E}_\mu[u_i(\veca)]
    \ge\max_{a_i'\in\Acal_i}u_i(a_i',\vecx_{-i})-\eta$ by applying $\eta$-CE obedience to constant deviation maps. Subtracting $u_i(\vecx)$, summing over both players, and using~\eqref{eq:rank-two-decomposition} gives the upper bound
\begin{align*}
    \sum_{i=1}^2\left(\max_{a_i'}u_i(a_i',\vecx_{-i})-u_i(\vecx)\right)
    &\le2\eta+\sum_{i=1}^2\left(\mathbb{E}_\mu[u_i(\veca)]-u_i(\vecx)\right)\\
    &=2\eta+\sum_{i=1}^2\operatorname{Cov}_\mu(f_i(a_1),g_i(a_2)).
\end{align*}
\end{proof}

In light of the previous lemma, it suffices to bound the covariance terms. This is done in the following lemma, crucially relying on the constant expectation property.

\ranktwoamplification*

\begin{proof}
We use the notation $u'_i(a_i)\defeq\mathbb{E}_\mu[u_i'(\veca)\mid a_i]$ for supported actions, and $\bar u_i\defeq\mathbb{E}_\mu[u_i'(\veca)]$. Since $\bar u_i =\sum_{a_i\in\supp(p_i^\mu)} p_i^\mu(a_i) u'_i(a_i)$, \Cref{def:constant-expectation} gives
\begin{equation}
\label{eq:averaged-centering}
    \sum_{a_i\in\supp(p_i^\mu)}p_i^\mu(a_i)| u'_i(a_i)-\bar u_i|
    \le\sum_{a_i\in\supp(p_i^\mu)}p_i^\mu(a_i)|u'_i(a_i)-v_i|
      +|v_i-\bar u_i|
    \le2\epsilon.
\end{equation}
Taking conditional and unconditional expectations in~\eqref{eq:rank-two-transformation}, and using $|u_1+2g_2|\le3$, gives
\[
    \left|\mathbb{E}_\mu[g_1(a_2)\mid a_1]-\mathbb{E}_\mu[g_1(a_2)]\right|
    \le4|u'_1(a_1)-\bar u_1|+\frac6K.
\]
Averaging over $a_1$ and applying~\eqref{eq:averaged-centering}, we obtain
\[
    \sum_{a_1\in\supp(p_1^\mu)}p_1^\mu(a_1)
    \left|\mathbb{E}_\mu[g_1(a_2)\mid a_1]-\mathbb{E}_\mu[g_1(a_2)]\right|
    \le8\epsilon+\frac6K.
\]
As a result, since $|f_1|\le1$,
\begin{align*}
    |\operatorname{Cov}_\mu(f_1(a_1),g_1(a_2))|
    &=\left|\mathbb{E}_\mu\left[f_1(a_1)
       \left(\mathbb{E}_\mu[g_1(a_2)\mid a_1]-\mathbb{E}_\mu[g_1(a_2)]\right)\right]\right|\\
    &\le8\epsilon+\frac6K.
\end{align*}
The bound with respect to Player~2 follows symetrically.
\end{proof}

\twoplayerconstantexpectationhardness*

\begin{proof}
Given a game $\mathcal{G}$ as above, we choose $\eta=2^{-p(L)}\in(0,1]$, where $L$ is its bit complexity and $p$ is a polynomial, sufficiently small that any rational $\eta$-Nash equilibrium can be converted in polynomial time into an exact Nash equilibrium; for example, this can be done through an LP~\citep{Etessami10:Complexity}. We set the parameters as $K \defeq 48/\eta$, $\epsilon \defeq \eta/32$, $\delta \defeq \eta^2/1536 = \eta/(32K)$, and construct $\mathcal{G}'$ using~\eqref{eq:rank-two-transformation}. The construction is polynomial in $L$.

Let $\mu$ be any $(\epsilon,\delta)$-constant-expectation correlated equilibrium of $\mathcal{G}'$, and let $\vecx$ be its marginal strategy profile. By~\Cref{lem:rank-two-transformation}, $\mu$ is a $(4K\delta)$-CE of $\mathcal{G}$. Therefore, \Cref{lem:rank-two-covariance,lem:rank-two-amplification} give
\[
    \sum_{i=1}^2\left(
        \max_{a_i'}u_i(a_i',\vecx_{-i})-u_i(\vecx)
    \right)
    \le16\epsilon+8K\delta+\frac{12}{K} = \eta,
\]
where the equality follows by our parameter choice. This proves that $\vecx$ is an $\eta$-Nash equilibrium, from which we can recover an exact Nash equilibrium of $\mathcal{G}$.
\end{proof}

\section{Hardness for independent state recommendations}
\label{sec:nfce-hardness}

In this section, we show that, in a sense, constant-expectation CEs are {\em necessary} to compute NFCEs in Markov games. Thus, applying our  hardness results (\Cref{thm:constant-expectation-hardness,thm:two-player-constant-expectation-hardness}), it will follow that computing an $\eps$-NFCE with independent state recommendations is \PPAD-hard for games with a constant number of players when $\eps$ is exponentially small, and for $n+1$ players when $\eps = \Theta(1/n)$.

\nfceTransfer*
\begin{proof}

Player $n+1$ will be called the {\em tester}. The Markov game is constructed as follows. At the first stage, a player $i \in [n]$ is selected uniformly at random. That player plays a game of matching pennies against the tester: if the pennies match, then $i$ gets $2$; else, the tester gets $2$. If the player played tails, the game ends. Otherwise, at stage $2$, the $n$ players (excluding the tester) play $\cG$ and gain the associated rewards. There are $n$ states at stage 1 and one state at stage 2, for a total of $n+1$ states. All unspecified rewards are zero, and inactive players have a single dummy action.

% Now consider any $\eps$-NFCE of $\cM$, where $\eps \le 1/8$. At the first stage, the selected player always plays heads with probability lower bounded by $1/4$; otherwise, the player and tester cannot be in an approximate equilibrium of that zero-sum game. Thus, $\cG$ is reached with probability at least $1/4$. Thus, 

% It remains to show that this correlated profile is also $\eps$-constant-expectation. 

Consider any $\eps$-NFCE of $\cM$ with independent state recommendations. (Throughout this proof, we use the \emph{ex-ante} deviation inequalities of an $\eps$-NFCE.) Let $\mu$ be the correlated profile recommended at $\cG$. We claim that $\mu$ is a $(4n\eps,4\eps)$-constant-expectation CE.
Let $v_i\defeq\mathbb{E}_{\veca\sim\mu}[u_i(\veca)]$. Independence across states implies that deviations affecting only the first-stage action, and ignoring the second-stage recommendation, have continuation value $v_i$ if $\cG$ is reached. If Player $i$ is selected, then Player $i$ and the tester play a game strategically equivalent to the zero-sum game with row utility matrix
\[
\renewcommand{\arraystretch}{1.25}
\setlength{\arraycolsep}{12pt}
\begin{array}{c|cc}
     & \mathsf{H}     & \mathsf{T} \\ \hline
\mathsf{H}    & \cellcolor{payoffgray} 2 + v_i & v_i \\
\mathsf{T}    & 0 & \cellcolor{payoffgray} 2
\end{array}
\]
where the player is the row player and the tester is the column player. %Indeed, subtracting $2+v_i\mathbf{1}\{\mathsf{H}\}$ from the tester's payoff makes its payoff the negative of the displayed matrix. 
% This term depends only on the row action, so the subtraction preserves the tester's deviation gains. 
The unique Nash equilibrium has the player playing heads with probability $1/2$, the tester playing heads with probability $(2-v_i)/4\in[1/4,1/2]$, and value $V_i\defeq 1 + v_i/2$ for $i$. Let $\tilde V_i$ be the row player's expected payoff conditional on reaching this state. 
% Let $U_i$ be the row player's expected payoff under the two first-stage recommendations, including the continuation value, and let 
The state is reached with probability $1/n$, so a deviation confined to this state has conditional gain at most $n\eps$. Therefore, we have
\[
    |V_i - \tilde V_i| \le n\eps.
\]
The tester's deviation to {\sf H} gives
\[
    \tilde V_i \le (2+v_i)q_i + n \eps,
\]
where $q_i$ is the probability that Player $i$ plays heads. Therefore,
\[
    q_i \ge \frac{\tilde V_i - n\eps}{2 + v_i} \ge \frac{V_i - 2n\eps}{2V_i} \ge \frac{1}{2} - n\eps \ge \frac{1}{4},
\]
where we use the inequalities $\tilde V_i \ge V_i - n\eps$, $V_i \ge 1$, and $n\eps \le 1/4$.
Therefore, $\cG$ is reached with probability at least $1/4$, and so $\mu$ is a $4\eps$-CE.

It remains to show the constant-expectation property. Let $p_i$ be the probability with which the tester is recommended to play heads against $i$. For every recommendation $a_i$ in the support of $p_i^\mu$, let
\[
    u_i(a_i)\defeq
    \mathbb{E}_{\veca\sim\mu}[u_i(\veca)\mid a_i].
\]
Consider a deviation that, only when $i$ is selected, ignores its first-stage recommendation and chooses heads or tails using only $a_i$, while obeying at $\cG$ if it is reached. By independence across states, heads has expected value $2p_i+u_i(a_i)$ and tails has expected value $2(1-p_i)$. Therefore, choosing the most profitable action and averaging over $a_i$ yields conditional utility
\begin{align*}
    V_i^* &= \sum_{a_i} p_i^\mu(a_i)
       \max\{2p_i+u_i(a_i),\ 2(1-p_i)\}
    =V_i+\frac12\sum_{a_i} p_i^\mu(a_i)
       |(2-4p_i)-u_i(a_i)|.
\end{align*}
This deviation cannot increase utility conditional on the selected state by more than $n\eps$. As a result, $V_i^*\le \tilde V_i+n\eps\le V_i+2n\eps$, and
\[
    \sum_{a_i\in\supp(p_i^\mu)} p_i^\mu(a_i)
    |(2-4p_i)-u_i(a_i)|\le 4n\eps.
\]
Therefore, choosing $v_i' \defeq 2 - 4 p_i$, and dividing all rewards by $4$ to satisfy the normalization assumption, completes the proof. We finally note that the construction preserves succinct payoff representations with polynomial overhead.
\end{proof}

\end{document}